\documentclass[12pt]{article}
\usepackage[margin=2.5cm]{geometry}
\usepackage{setspace}
\usepackage{amsmath, amssymb, amsthm, mathrsfs}
\usepackage{bbm}
\usepackage{graphicx}
\usepackage{xcolor}
\usepackage{tikz}
\usepackage{pgfplots}
\usepackage{booktabs}
\usepackage{float}
\usepackage[numbers,authoryear]{natbib}
\setcitestyle{authoryear,open={(},close={)}}
\usepackage[hidelinks]{hyperref}
\theoremstyle{plain}
\newtheorem{theorem}{Theorem}
\newtheorem{proposition}{Proposition}
\newtheorem{lemma}{Lemma}
\newtheorem{corollary}{Corollary}
\theoremstyle{definition}
\newtheorem{definition}{Definition}
\newtheorem{axiom}{Axiom}
\newtheorem{example}{Example}
\theoremstyle{remark}
\newtheorem{remark}{Remark}

\definecolor{lightcoral}{rgb}{0.94,0.5,0.5}
\definecolor{indianred}{rgb}{0.8,0.36,0.36}
\definecolor{firebrick}{rgb}{0.7,0.13,0.13}
\definecolor{lightblue}{rgb}{0.68,0.85,0.9}
\definecolor{cornflowerblue}{rgb}{0.39,0.58,0.93}
\definecolor{royalblue}{rgb}{0.25,0.41,0.88}

\title{Integrating Proportionality and Egalitarianism in Claims Problems\thanks{We are grateful to William Thomson, Hervé Moulin, Clemens Puppe, Juan Moreno-Ternero, Ruben Juarez, and Alan Miller for their insightful comments and valuable suggestions. We also thank participants at SSCW 2024 – the 17th Meeting of the Society for Social Choice and Welfare, the Workshop on Social Choice, Game Theory, and Mechanism Design organized in honor of Hervé Moulin’s 75th birthday at the University of Glasgow, as well as seminar audiences at California State University, East Bay and the University of Texas Rio Grande Valley for helpful feedback and discussions.}}

\author{Anisha Bandyopadhyay \thanks{Department of Humanities and Social Sciences, Indian Institute of Technology Delhi, Hauz Khas, New Delhi 110016, India, E\textendash mail: aanishabanerjee@gmail.com} \and Sinan Ertemel \thanks{Department of Economics, Istanbul Technical University, 34367, Ma\c{c}ka, Istanbul, Türkiye, E\textendash mail: ertemels@itu.edu.tr} \and Rajnish Kumar \thanks{Queen's Business School, Queen's University Belfast, 185 Stranmillis Road, Belfast, BT9 5EE, United Kingdom, E\textendash mail: rajnish.kumar@qub.ac.uk}\and Saptarshi Mukherjee \thanks{Department of Humanities and Social Sciences, Indian Institute of Technology Delhi, Hauz Khas, New Delhi 110016, India, E\textendash mail: saptarshi.isi@gmail.com}}

\date{\today}
\begin{document}
\onehalfspacing
%\doublespacing

\maketitle

\begin{abstract}

\noindent We study the problem of allocating a finite estate among agents whose total claims exceed the available resources, a standard framework in the theory of claims problems. Two canonical rules embody competing fairness ideals: the Proportional rule allocates in proportion to claims, while the Constrained Equal Awards (CEA) rule equalizes awards as much as possible subject to claim-boundedness. We introduce the P--CEA family of compromise rules, which assigns each agent a fixed baseline award—capped by her claim—and distributes the remaining estate proportionally to residual claims. By varying the baseline parameter, this family generates a continuum of allocation rules that interpolates between the Proportional and CEA benchmarks. We provide an axiomatic characterization based on two threshold-dependent principles: \emph{No Advantageous Reallocation}, which prevents agents with claims above the threshold from benefiting through coordinated claim redistribution that preserves the threshold condition, and \emph{Sustainable Lower Bound}, which guarantees each agent at least the minimum of her claim and the threshold. We further develop a dual analysis that reallocates losses instead of awards and characterize the corresponding dual family using the dual versions of our axioms.

\vspace{1em}

\vspace{0.5em}

\end{abstract}

\section{Introduction}

The problem of dividing a finite estate among agents whose total claims exceed the available resources is known as a \emph{claims problem}. This setting occupies a central place in fair division theory, which seeks allocation rules that are both normatively compelling and operationally feasible. Among the most prominent division rules are the \emph{Proportional rule} and the \emph{Constrained Equal Awards (CEA) rule}. The proportional rule allocates resources in proportion to claims, reflecting a merit-based view of fairness. The CEA rule, in contrast, equalizes awards as far as possible subject to claim-boundedness, embodying an egalitarian ideal. While each rule captures a powerful normative intuition, both have well-known limitations: proportionality may disadvantage agents with small claims, whereas CEA may insufficiently account for differences in entitlement.

A natural question, therefore, is whether one can design principled compromise rules that combine the strengths of both approaches. In this paper, we propose such a family of rules, which we call the \emph{P--CEA family}. Each rule in this family assigns every agent a fixed, claim-bounded baseline amount and distributes the remaining estate proportionally over residual claims. By varying the baseline parameter, the family generates a continuum of allocation rules ranging from the proportional rule to the CEA rule. This structure provides a flexible and transparent way to balance proportionality and egalitarianism.

Beyond its theoretical appeal, this compromise structure has direct relevance in contemporary policy contexts. A salient example arises in the allocation of the limited global carbon budget among countries. Climate mitigation can naturally be framed as a “carbon budget bankruptcy” problem \citep{gimenez2016global}: the finite carbon budget represents the estate, while countries’ projected emissions represent competing claims. Under this interpretation, allocation rules from the P--CEA family offer a transparent mechanism for distributing the remaining carbon budget. For instance, countries’ claims may reflect their Intended Nationally Determined Contributions (INDCs) under the Paris Agreement, informed by emissions scenarios such as those developed by the Intergovernmental Panel on Climate Change \citep{ipcc2014ipcc}. Given a global carbon budget consistent with limiting warming to below 2$^\circ$C \citep{meinshausen2009greenhouse}, a rule from our family guarantees each country a baseline allocation while distributing the remainder proportionally to its residual claim. Such a mechanism explicitly combines a universal minimum with recognition of differential demands.

The same logic appears in a variety of allocation environments. In the English Premier League, broadcasting revenues are divided between an equal-share component and performance-based merit payments. The equal share guarantees each club a baseline amount, while merit payments reward relative success. Similarly, quota systems in fisheries or dairy production often ensure producers a survival allocation before distributing additional quantities according to capacity or historical production. Food rationing systems and the allocation of essential goods—such as water or medical supplies—frequently follow a two-stage logic: secure a minimum level for all recipients, then allocate remaining resources in proportion to need or population. These arrangements reflect an intuitive principle: fairness may require guaranteeing a floor before accounting for differences in claims.

A related philosophical parallel arises in the idea of Universal Basic Income (UBI). UBI embodies the principle that all individuals are entitled to a foundational level of economic security independent of status or productivity. This mirrors the core structure of the P--CEA family: a fixed baseline allocation supplemented by a component responsive to claims or entitlements. Real-world examples, such as the Alaska Permanent Fund and large-scale UBI pilot programs, demonstrate the practical relevance of embedding a guaranteed minimum within broader distributive systems.

Together, these examples illustrate the broad normative and practical appeal of allocation rules that integrate a universal guarantee with proportional differentiation. The P--CEA family formalizes this compromise within the theory of claims problems and provides a rigorous foundation for its analysis.

\subsection{Overview of Our Results}

We introduce a parametric family of allocation rules that assigns each agent a fixed, exogenously specified minimum award—subject to claim-boundedness—and distributes the remaining estate proportionally over residual claims. By varying the minimum parameter, the family generates a continuum of rules ranging from the proportional rule to the constrained equal awards (CEA) rule. Each member of the family thus represents a principled compromise between proportionality and egalitarianism.

Our first main contribution is an axiomatic characterization of this family. The characterization is based on two fairness principles. The first, \emph{No Advantageous Reallocation Beyond the Minimum} ($NAR_L$), requires that agents whose claims exceed the prescribed baseline cannot improve their collective allocation by redistributing claims among themselves while keeping their aggregate claim fixed. This condition limits strategic manipulation and protects agents with smaller claims from being adversely affected by coordinated misreporting among higher-claim agents. The second axiom, \emph{Sustainable Lower Bound on Awards} ($SLBA_L$), ensures that each agent receives at least the minimum of her claim and the specified threshold $L$, thereby embedding a guaranteed floor into the allocation process. Together, these axioms uniquely characterize the P--CEA family.

Our second contribution shows that $NAR_L$ admits an equivalent formulation in terms of \emph{Non-subsistence Decentralizability}. This alternative axiom requires that, once the subsistence (below-threshold) claims are fixed, each agent’s award depends only on her own claim, the aggregate claims, and the available estate. In other words, an agent need not know the entire profile of claims to determine her allocation. We show that the combination of decentralizability and the sustainable lower bound axiom also characterizes the P--CEA family, providing an alternative and structurally appealing foundation for the rules.

Finally, we develop the dual analysis of the model. We introduce dual counterparts of the lower-bound axiom and establish parallel characterizations for the associated loss-based rules. These dual rules mirror the structure of the original family but operate by equalizing losses rather than awards, revealing a symmetric structure underlying the compromise between proportionality and egalitarianism.

\subsection{Relevant Literature}

The study of fairness in the allocation of scarce resources has deep intellectual roots, spanning philosophical, religious, and mathematical traditions. One of the earliest formal treatments of conflicting claims problems is due to \citet{o1982problem}, who drew on Talmudic arbitration principles to analyze rights-based divisions. This approach was subsequently developed by \citet{aumann1985game}, who introduced a cooperative game-theoretic framework for bankruptcy problems grounded in classical legal doctrines. \citet{curiel1987bankruptcy} further formalized the structure of bankruptcy games and examined systematic methods for allocating resources among claimants.

A substantial strand of the literature investigates compromise solutions between the proportional rule and the constrained equal awards (CEA) rule. \citet{thomson2007existence} studies award paths that combine a segment along the 45-degree line with a segment leading to the claims vector, and explores consistent extensions of such constructions to multi-agent settings. In that framework, however, the equal minimum assigned to each agent does not exceed the smallest claim, whereas in our model the guaranteed baseline is exogenously specified and may differ from the lowest claim.

Focusing initially on the two-agent case, \citet{thomson2015a} propose a graphical compromise between the proportional and CEA rules and characterize it using claims continuity and composition down. Their construction averages allocations along an axis parallel to the larger claim and then extends the rule to multiple agents via consistency. Our approach differs in motivation and structure: rather than averaging benchmark allocations, we impose a guaranteed minimum award for each agent and allocate the remainder proportionally. Geometrically, each rule in our family consists of an initial segment along the 45-degree line—whose length depends on the chosen parameter—followed by a segment with the slope of the proportional rule.

In a related contribution, \citet{thomson2015b} consider compromise rules defined as weighted averages of proportional and CEA allocations, with weights depending on the claims profile. In the multi-agent case, however, only the two benchmark rules themselves satisfy consistency within that family. By contrast, every member of the P--CEA family satisfies consistency, preserving a key structural property of classical bankruptcy rules.

The rule most closely related to ours is the $\alpha_{\min}$ rule introduced by \citet{gimenez2014proportional}. This rule assigns the largest feasible equal award—equal to the smallest claim—and distributes the remaining estate proportionally. It can thus be interpreted as a convex combination of equal awards and proportionality. The P--CEA family shares this two-stage structure but differs in two important respects. First, the baseline level in our framework is exogenously chosen rather than determined by the smallest claim. Second, unlike the $\alpha_{\min}$ rule,  the P--CEA family satisfies consistency. Third, $\alpha_{\min}$ rule coincides with CEA in the two-agent case while $P-CEA$ family of rules preserves a genuine compromise structure even for two agents.

Our axiomatic characterizations build on variants of foundational principles in the bankruptcy literature, including no advantageous reallocation \citep{moulin1985egalitarianism}, feasible lower bounds \citep{moulin1992welfare}, and decentralizability \citep{moulin1991axioms}. For comprehensive surveys of claims problems and fair division, see \citet{moulin2002axiomatic} and \citet{thomson2019}.

More recently, the bankruptcy framework has been applied to global policy challenges, particularly in environmental contexts. \citet{gimenez2016global} interpret the global carbon budget as a conflicting claims problem, highlighting the normative role of allocation rules in climate negotiations. \citet{duro2020allocation} analyze the distribution of CO$_2$ emissions under alternative fairness criteria, and \citet{ju2021fair} develop cooperative approaches to emissions abatement. These contributions underscore the growing relevance of axiomatic fairness principles in global resource governance and motivate the study of flexible compromise rules such as the P--CEA family.

\subsection{Structure of the Paper}

The remainder of the paper is organized as follows. Section~2 presents the formal model and introduces the P--CEA family of compromise rules. We discuss their structural properties, examine the fairness principles they satisfy, and analyze their performance under standard inequality criteria. Section~3 develops the axiomatic framework and provides characterizations of the P--CEA family based on our proposed principles. Section~4 turns to the dual perspective, introducing the corresponding loss-based rules and establishing parallel characterizations using dual axioms. Section~5 concludes with a discussion of the main findings and outlines directions for future research.

\section{Model}

In this section, we introduce the standard framework and formally define the family of compromise rules. We then examine the fairness principles these rules satisfy and evaluate their performance under selected measures of inequality.

Consider a set of agents $N = \{1, 2, \dots, n\}$, where each agent $i \in N$ has a claim $c_i$ on an estate $e$. Let $\mathbf{c} = (c_1, c_2, \ldots, c_n)$ denote the vector of claims. Without loss of generality, we assume that the claims are ordered in ascending fashion, that is, $c_1 \leq c_2 \leq \cdots \leq c_n$.

If the available estate is sufficient to satisfy all claims, the allocation problem is straightforward. The problem becomes non-trivial when the aggregate claim $c_N = \sum_{i \in N} c_i$ exceeds the available estate, that is, when $c_N > e$. In such a situation, it is impossible to satisfy all claims in full; this setting is known as a \emph{claims problem}. A claims problem is therefore represented by the ordered pair $(\mathbf{c}, e)$.

Let $\mathcal{C}$ denote the set of all such claims problems. For each $(\mathbf{c}, e) \in \mathcal{C}$, a \emph{division rule} assigns an allocation of the estate among the agents. Formally, a division rule is a mapping $\varphi: \mathcal{C} \rightarrow \mathbb{R}^n_+$. Let $\mathbf{X}(\mathbf{c}, e)$ denote the set of admissible award vectors for the problem $(\mathbf{c}, e)$.

Any division rule must satisfy three fundamental conditions:
(i) \emph{Non-negativity}: $0 \leq \varphi_i(\mathbf{c}, e)$ for all $i \in N$;
(ii) \emph{Claim-boundedness}: $\varphi_i(\mathbf{c}, e) \leq c_i$ for all $i \in N$;
(iii) \emph{Budget-balancedness}: $\sum_{i \in N} \varphi_i(\mathbf{c}, e) = e$.

Within these feasibility constraints, different division rules embody different fairness principles, which we introduce below.

\subsection{Benchmark Rules}

The family of rules proposed in this paper provides a compromise between two classical benchmarks in the fair division literature. We briefly review these foundational rules.

In \emph{Nicomachean Ethics}, Book V, Aristotle formulates the principle that “equals should be treated equally, and unequals unequally in proportion to their relevant similarities and differences.” This idea underlies the \emph{Proportional Rule (P)}, according to which the estate is distributed in proportion to individual claims:
\[
P(\mathbf{c}, e) \equiv \left(\tfrac{e}{c_N}\right)\mathbf{c}.
\]

A contrasting principle is equality before the law, emphasized in the \emph{Mishneh Torah}. This perspective motivates the \emph{Equal Awards (EA) method}, which allocates the estate equally among all agents:
\[
EA(\mathbf{c}, e) \equiv \left(\tfrac{e}{n}\right)\mathbf{1}_n,
\]
where $\mathbf{1}_n = (1, \ldots, 1)$ denotes the $n$-dimensional vector of ones.

Although normatively appealing, the EA method disregards the magnitudes of individual claims and may therefore violate claim-boundedness. The \emph{Constrained Equal Awards (CEA) Rule} addresses this limitation by equalizing awards as much as possible without allowing any agent to receive more than her claim. Formally,
\[
CEA(\mathbf{c}, e) \equiv \big(\min\{c_i, \lambda\}\big)_{i \in N},
\]
where $\lambda \in \mathbb{R}_+$ is chosen such that $\sum_{i \in N} \min\{c_i, \lambda\} = e$. The CEA rule thus produces the most equal allocation compatible with claim-boundedness.

The dual of the CEA rule is the \emph{Constrained Equal Losses (CEL) Rule}, which equalizes losses rather than awards. We discuss duality in detail in Section~4. The CEL rule is defined as
\[
CEL(\mathbf{c}, e) \equiv \big(\max\{c_i - \mu, 0\}\big)_{i \in N},
\]
where $\mu \in \mathbb{R}_+$ is chosen to satisfy $\sum_{i \in N} \max\{c_i - \mu, 0\} = e$.

Under claim-boundedness, the CEA rule yields the most egalitarian feasible division and typically favors agents with smaller claims. In contrast, the proportional rule preserves relative claim magnitudes and may disadvantage agents with small claims when disparities are large. The CEL rule lies at the opposite extreme, favoring agents with larger claims.

Thus, the proportional and CEA rules represent two endpoints of a fairness spectrum: the former emphasizes proportional entitlement, whereas the latter prioritizes egalitarian protection. The family of compromise rules introduced below provides a systematic way to interpolate between these two principles, allowing for more balanced and context-sensitive allocations.

\subsection{A Family of Compromise Rules}

We introduce the \emph{P--CEA family of compromise rules}, a class of allocation rules that interpolates between the Proportional rule and the Constrained Equal Awards (CEA) rule. Each rule in this family is indexed by a parameter $L \in \mathbb{R}_+$.

The construction proceeds in two steps. First, a fixed amount $L$ is assigned to each agent—capped at her claim. Second, the remaining estate is distributed proportionally among the residual claims. When $L = 0$, the rule coincides with the proportional rule; when $L = \lambda$, it coincides with the CEA rule. Intermediate values $L \in (0, \lambda)$ generate compromise rules that lie strictly between these two benchmarks.

Since, for a given estate, the maximum uniform amount that can be assigned to all agents without violating feasibility is $\lambda$, the effective range of the parameter is $[0,\lambda]$. Thus, while $L$ is chosen exogenously, its admissible range is determined endogenously by the estate and the claims vector. Hence, throughout the rest of the paper, by $L$ we will mean $\min\{L,\lambda\}$.

Suppose $c_k \leq L \leq c_{k+1}$ for some $k \in \{0, \ldots, n\}$. Define the vector of fixed awards as
\[
\mathbf{c}^L = \big(\min\{c_i, L\}\big)_{i \in N},
\]
and let the residual estate after assigning the fixed component be
\[
e^L = e - \sum_{i \in N} c_i^L.
\]
Let us also define by $S_0=\{1,2,...,k\}$ the set of agents whose claims are below the value of $L$. Also, we call $N^1=N\setminus S_0$ the set of agents whose claims are above the value L.
We now define the P--CEA family of rules.

\begin{definition}[The P--CEA Family of Compromise Rules]
For each $(\mathbf{c}, e) \in \mathcal{C}$ with $c_i \geq 0$ for all $i \in N$,
\[
\psi^L(\mathbf{c}, e) = \mathbf{c}^L + P(\mathbf{c} - \mathbf{c}^L, e^L).
\]
\end{definition}

We refer to $\psi^L(\mathbf{c}, e)$ as the \emph{P--CEA compromise rule with parameter $L$}.

Note that the case $L < c_1$ is permitted. In this situation, the fixed award vector becomes $\mathbf{c}^L = (L)_{i \in N}$, and the residual estate reduces to $e^L = e - nL$.

\begin{example}

Consider a two-agent problem with total estate $e = 100$ and claims 
$c_1 = 50$ and $c_2 = 100$.  
Under the CEA rule, the allocation is $(50, 50)$; under the 
proportional rule, it is $\left(\tfrac{100}{3}, \tfrac{200}{3}\right)$.  
The award paths generated by these two rules form the boundaries of a triangle 
in the award space, as illustrated in Figure~\ref{fig:award-paths}.  
The award paths corresponding to the P--CEA rules with $L \in (0, \lambda)$ 
lie strictly inside this triangle.

For instance, let $L = 25$. The fixed award vector is 
$\mathbf{c}^L = (\min\{50, 25\}, \min\{100, 25\}) = (25, 25)$, 
and the residual estate equals 
$e^L = 100 - (25 + 25) = 50$.  
Applying the proportional rule to the residual claims yields 
$P((25, 75), 50) = (12.5, 37.5)$.  
Hence, the final allocation under the P--CEA rule with $L = 25$ is 
$\psi^{25}(\mathbf{c}, e) = (25, 25) + (12.5, 37.5) = (37.5, 62.5)$.  
The corresponding award path is shown in black in Figure~\ref{fig:award-paths}.

\begin{figure}[H]
    \centering
    \begin{tikzpicture}
    \begin{axis}[
        axis x line=middle,
        axis y line=middle,
        xlabel = {First agent},
        ylabel = {Second agent},
        xmin=0, xmax=110,
        ymin=0, ymax=110,
        xtick={0,25,50,100},
        ytick={0,25,50,100},
        ticklabel style={font=\small},
        xlabel style={at={(ticklabel* cs:0.5)}, anchor=north, yshift=-16pt},
        ylabel style={at={(ticklabel* cs:0.7)}, anchor=east, rotate=90, yshift=30pt},
        legend pos=outer north east,
        clip=false
    ]
    \addplot[color=blue, mark=o]
        coordinates {(0,0)(33.33,66.66)(50,100)};
    \addplot[color=red, mark=o]
        coordinates {(0,0)(50,50)(50,100)};
    \addplot[color=black, mark=o]
        coordinates {(0,0)(25,25)(37.5,62.5)(50,100)};
    \draw[dotted] (25,25) -- (25,0);
    \draw[dotted] (25,25) -- (0,25);
    \draw[dotted] (50,50) -- (50,0);
    \draw[dotted] (50,50) -- (0,50);
    \draw[dotted] (50,100) -- (0,100);
    \draw[brown] (100,0) -- (0,100);
    \node[right] at (50,100){$(c_1,c_2)$};
    \node[below left] at (axis cs:0,0) {$0$};
    \legend{Proportional rule, CEA rule, P--CEA rule with $L=25$}
    \end{axis}
    \end{tikzpicture}
    \caption{Award paths under the proportional rule, 
    the CEA rule, and the P--CEA rule with $L = 25$, for two agents.}
    \label{fig:award-paths}
\end{figure}
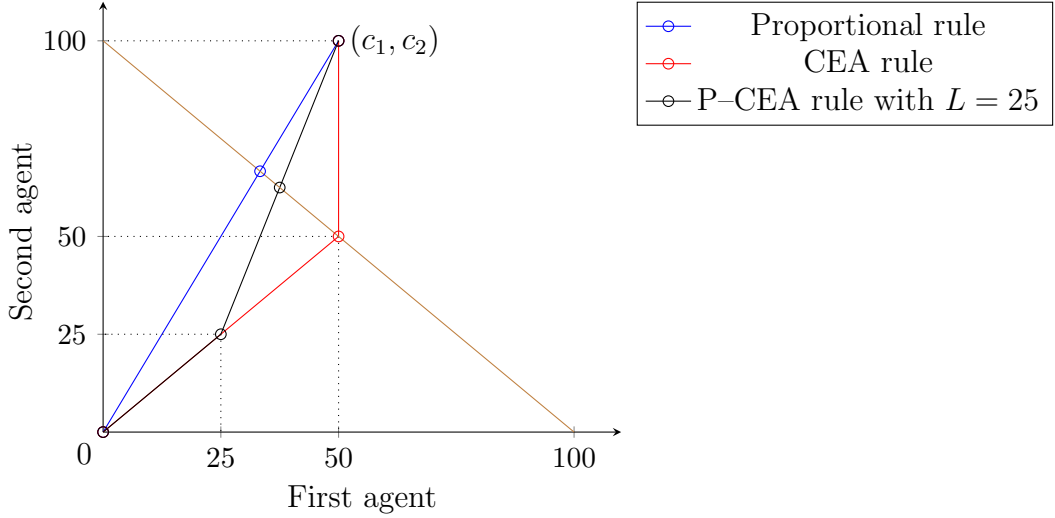

\end{example}
A division rule that solves a claims problem is generally expected to satisfy a number of fundamental fairness criteria. Below, we review several well-established principles and show that the proposed $\psi^L$ family satisfies each of them.

\noindent\textbf{Continuity.}
For each sequence $\{(\mathbf{c}^v,e^v)\}$ in $\mathcal{C}$ and each $(\mathbf{c},e)\in \mathcal{C}$, if $(\mathbf{c}^v,e^v)\rightarrow (\mathbf{c},e)$, then $\varphi(\mathbf{c}^v,e^v)\rightarrow \varphi(\mathbf{c},e)$.  
This property requires that small changes in the claims problem lead only to small changes in the allocation.  
The $\psi^L$ family satisfies continuity.

\noindent\textbf{Equal Treatment of Equals.}  
For each $N \in \mathbb{N}$, each $(\mathbf{c}, e) \in \mathcal{C}$, and each pair $\{i, j\} \subseteq N$, if $c_i = c_j$, then $\varphi_i(\mathbf{c}, e) = \varphi_j(\mathbf{c}, e)$.  
Thus, agents with identical claims must receive identical awards.  
Every rule in the $\psi^L$ family satisfies this property.

\noindent\textbf{Order Preservation.}  
For each $N \in \mathbb{N}$, each $(\mathbf{c}, e) \in \mathcal{C}$, and for any $i, j \in N$ such that $c_i \geq c_j$, we have
\[
\varphi_i(\mathbf{c}, e) \geq \varphi_j(\mathbf{c}, e)
\quad \text{and} \quad
c_i - \varphi_i(\mathbf{c}, e) \geq c_j - \varphi_j(\mathbf{c}, e).
\]
This condition ensures that larger claims result in weakly larger awards and weakly larger residual losses.  
All $\psi^L$ rules satisfy order preservation.

\noindent\textbf{Resource Monotonicity.}  
For each $N \in \mathbb{N}$, each $(\mathbf{c}, e) \in \mathcal{C}$, and for $e' > e$ with $c_N > e'$, we have
\[
\varphi_i(\mathbf{c}, e') \geq \varphi_i(\mathbf{c}, e),
\quad \text{for all } i \in N.
\]
Thus, an increase in the available estate cannot reduce any agent’s award.  
Every rule in the $\psi^L$ family satisfies resource monotonicity.

\noindent\textbf{Anonymity.}  
For each pair $\{N, N'\} \subset \mathbb{N}$ with $|N| = |N'|$, each $(\mathbf{c}, e) \in \mathcal{C}$, each bijection $\pi: N \to N'$, and each $i \in N$, we have
\[
\varphi_{\pi(i)}\big((c_{\pi(i)})_{i \in N}, e\big)
=
\varphi_i(\mathbf{c}, e).
\]
This property requires that the allocation depend only on claims and not on agent identities.  
The $\psi^L$ family satisfies anonymity.

\noindent\textbf{Consistency.}  
For each $N \in \mathbb{N}$, each $(\mathbf{c}, e) \in \mathcal{C}$, and each $N' \subset N$, if $x \equiv \varphi(\mathbf{c}, e)$, then
\[
\varphi\left(\mathbf{c}|_{N'}, \sum_{i \in N'} x_i\right)
=
\mathbf{x}|_{N'},
\]
where $\mathbf{c}|_{N'} := \{c_i\}_{i \in N'}$ and $\mathbf{x}|_{N'} := \{x_i\}_{i \in N'}$.  
This property requires that, when restricting attention to a subset of agents and their allocated total, reapplying the rule reproduces their original awards.  

Each $\psi^L$ rule consists of a fixed component determined by $L$ and a proportional component. The fixed component remains unchanged in subproblems, and the proportional rule is consistent by construction.  
Therefore, the $\psi^L$ family satisfies consistency.

\noindent\textbf{Converse Consistency.}  
For each $N \in \mathbb{N}$, each $(\mathbf{c}, e) \in \mathcal{C}$, and each $x \in \mathbf{X}(\mathbf{c}, e)$, if for every pair $N' \subset N$ with $|N'| = 2$ we have
\[
\mathbf{x}|_{N'}
=
\varphi\left(\mathbf{c}|_{N'}, \sum_{i \in N'} x_i\right),
\]
then $x = \varphi(\mathbf{c}, e)$.  
This condition requires that if an allocation agrees with the rule on all two-agent subproblems, then it must coincide with the rule on the full problem.  

Since the proportional component of $\psi^L$ satisfies converse consistency and the fixed component is independent of the estate level, the entire $\psi^L$ family satisfies converse consistency.

\noindent\textbf{Composition Down.}  
For each $(\mathbf{c}, e) \in \mathcal{C}$ and for every $e' < e$, we require
\[
\varphi(\mathbf{c}, e') = \varphi\big(\varphi(\mathbf{c}, e), e'\big).
\]
This property states that if the total estate is reduced after an initial allocation, then applying the rule to the reduced estate and the previously allocated vector yields the same outcome as applying the rule directly to the original claims with the smaller estate.

We show that the $\psi^L$ family satisfies composition down. As the estate decreases from $e$ to $e'$, the effective range of $L$ may also shrink. Let the admissible range of $L$ be $[0,\lambda]$ when the estate is $e$, and $[0,\lambda']$ when the estate is $e'$. Observe that $\lambda' \le \lambda$. Depending on the relative position of $L$ with respect to $\lambda$ and $\lambda'$, we distinguish two cases.

\emph{Case 1:} $L \le \lambda' < \lambda$.

In this case,
$\psi^L(\mathbf{c}, e) = \mathbf{c}^L + P(\mathbf{c}-\mathbf{c}^L, e^L)$
and
$\psi^L(\mathbf{c}, e') = \mathbf{c}^L + P(\mathbf{c}-\mathbf{c}^L, e^{\prime L})$.
Since the fixed component $\mathbf{c}^L$ remains unchanged and the proportional rule satisfies composition down, the result follows immediately.

\emph{Case 2:} $\lambda' \le L \le \lambda$, with at least one strict inequality.

In this case,
$\psi^L(\mathbf{c}, e) = \mathbf{c}^L + P(\mathbf{c}-\mathbf{c}^L, e^L)$,
whereas
$\psi^L(\mathbf{c}, e') = \mathbf{c}^{\lambda'} + P(\mathbf{c}-\mathbf{c}^{\lambda'}, e^{\lambda'})$,
with $\mathbf{c}^{\lambda'} = (\min\{c_i,\lambda'\})_{i \in N}$ and
$e^{\lambda'} = e' - \sum_{i \in N} c^{\lambda'}$.

Since $\min\{c_i,\lambda'\} \le \min\{c_i,L\}$ for all $i$, we have
$\sum_{i \in N} c^{\lambda'} \le \sum_{i \in N} c^L$. Thus, the reduction in the fixed component
$\sum_{i \in N} c^L_i - \sum_{i \in N} c^{\lambda'}_i$ is transferred to the proportional component.
Accordingly,
\[
\begin{aligned}
\psi\!\left(\mathbf{c}^L + P(\mathbf{c}-\mathbf{c}^L, e^L),\, e' \right)
&= \mathbf{c}^{\lambda'} 
   + P\!\left(\mathbf{c}-\mathbf{c}^{\lambda'},\, 
     e' - \sum_{i \in N} \mathbf{c}^L + (\sum_{i \in N} c^L_i - \sum_{i \in N} c^{\lambda'}_i) \right) \\
&= \mathbf{c}^{\lambda'}
   + P\!\left(\mathbf{c}-\mathbf{c}^{\lambda'},\, e^{\lambda'} \right).
\end{aligned}
\]
Hence, every rule in the P--CEA family satisfies composition down.

Composition down implies endowment monotonicity. As shown in Theorem~6.1 of \cite{thomson2019}, any family of rules satisfying this property induces a structure of \emph{award paths} that can be represented as monotone, space-filling trees in the award space. Each path originates at the origin and expands outward as the estate increases. Once two paths intersect, they coincide thereafter. These paths therefore form a tree-like structure (see Figure~\ref{fig:monotone-tree}) that evolves smoothly as the parameter $L$ varies.

\begin{figure}[H]
    \centering

    % --- (a) Fixed L ---
    \begin{minipage}[b]{0.7\textwidth}
        \centering
        \begin{tikzpicture}
            \begin{axis}[
                width=8cm, height=7cm,
                axis x line=middle,
                axis y line=middle,
                xlabel={First agent},
                ylabel={Second agent},
                xmin=0, xmax=50,
                ymin=0, ymax=100,
                ticklabel style={font=\small},
                xlabel style={at={(ticklabel* cs:0.6)}, anchor=north, yshift=-15pt},
                ylabel style={at={(ticklabel* cs:0.7)}, anchor=east, rotate=90, yshift=25pt},
                clip=false
            ]
            \addplot+[mark=o, color=lightcoral, thick] coordinates {(0,0) (10,10) (10,50)};
            \addplot+[mark=o, color=indianred, thick] coordinates {(0,0) (30,30) (32,70)};
            \addplot+[mark=o, color=firebrick, thick] coordinates {(0,0) (30,30) (40,60)};
            \end{axis}
        \end{tikzpicture}
        
        \vspace{2mm}
        
        \textbf{(a)} Award paths for varying claims, fixed $L=30$.
    \end{minipage}

    \vspace{1cm}

    % --- (b) Changing L ---
    \begin{minipage}[b]{0.7\textwidth}
        \centering
        \begin{tikzpicture}
            \begin{axis}[
                width=8cm, height=7cm,
                axis x line=middle,
                axis y line=middle,
                xlabel={First agent},
                ylabel={Second agent},
                xmin=0, xmax=50,
                ymin=0, ymax=100,
                xtick={0,10,20,30,40,50},
                ytick={0,20,40,60,80,100},
                ticklabel style={font=\small},
                xlabel style={at={(ticklabel* cs:0.6)}, anchor=north, yshift=-15pt},
                ylabel style={at={(ticklabel* cs:0.7)}, anchor=east, rotate=90, yshift=25pt},
                clip=false
            ]
            \addplot+[mark=o, color=lightcoral, thick] coordinates {(0,0) (10,10) (10,50)};
            \addplot+[mark=o, color=indianred, thick] coordinates {(0,0) (30,30) (32,70)};
            \addplot+[mark=o, color=firebrick, thick] coordinates {(0,0) (30,30) (40,60)};
            \addplot+[mark=o, dashed, color=lightblue, thick] coordinates {(0,0) (5,5) (10,50)};
            \addplot+[mark=o, dashed, color=cornflowerblue, thick] coordinates {(0,0) (5,5) (32,70)};
            \addplot+[mark=o, dashed, color=royalblue, thick] coordinates {(0,0) (5,5) (40,60)};
            \end{axis}
        \end{tikzpicture}
        
        \vspace{2mm}
        
        \textbf{(b)} Tree structure evolution as $L$ decreases from 30 to 5.
    \end{minipage}

    \caption{Monotone tree structures induced by the $\psi^L$ family in award space.}
    \label{fig:monotone-tree}
\end{figure}

When the rule additionally satisfies \emph{homogeneity}, the award space can be partitioned into cones, each generated by a curve together with its homothetic transformations. In the two-agent case, every generating curve—except possibly for an initial segment—is visible from the origin. In our family, visibility is preserved because all non-degenerate segments pass through the origin, thereby maintaining the tree structure that is compatible with homogeneity.

\noindent\textbf{Composition Up.}  
For each $(\mathbf{c}, e) \in \mathcal{C}$ and each $e' > e$ such that $c_N > e'$, we require
\[
\varphi(\mathbf{c}, e') = \varphi(\mathbf{c}, e) + \varphi\big(\mathbf{c} - \varphi(\mathbf{c}, e),\, e' - e\big).
\]
This property states that when the estate increases, the additional amount can be allocated by applying the rule to the residual claims and adding the resulting allocation to the original awards. Since the proportional rule satisfies composition up and the fixed component of $\psi^L$ does not depend on the estate level once $L$ is fixed, the $\psi^L$ family also satisfies this property.

\noindent\textbf{Young Rule.}  
Let $\Phi$ denote the family of functions $f: \mathbb{R}_+ \times [\underline{\beta}, \bar{\beta}] \to \mathbb{R}_+$, where $-\infty \le \underline{\beta} \le \bar{\beta} \le \infty$, such that:  
(i) $f$ is continuous and non-decreasing in its second argument;  
(ii) for each $c_0 \in \mathbb{R}_+$, $f(c_0, \underline{\beta}) = 0$ and $f(c_0, \bar{\beta}) = c_0$.

For a claims problem $(\mathbf{c}, e) \in \mathcal{C}$, the \emph{Young rule} associated with $f \in \Phi$ is defined by
\[
Y^f(\mathbf{c}, e) = \big(f(c_i, \beta)\big)_{i \in N},
\]
where $\beta \in [\underline{\beta}, \bar{\beta}]$ is chosen so that $\sum_{i \in N} f(c_i, \beta) = e$.

The proposed $\psi^L$ family constitutes a subclass of the Young rules. For a fixed parameter $L$, the associated function can be written as
\begin{equation*}
    f(c_i, \beta; L) = 
    \begin{cases}
        \min\{c_i, \beta\} &\text{ when } \beta \leq L \\
        \min\{c_i, L\} + (c_i - \min\{c_i, L\})(\beta-L) &\text{ when } L<\beta\leq L+1
    \end{cases}
\end{equation*}
Here, $\underline{\beta}=0$ and $\overline{\beta}=1+L$. This function satisfies:
\begin{itemize}
\item $f(c_i,\underline{\beta} ) = 0$, the minimum award assigned to an agent;
\item $f(c_i, \overline{\beta}) = c_i$, the maximum feasible award;
\item $f$ is continuous and non-decreasing in $\beta$.
\end{itemize}

Alternatively, the conclusion follows from the fact that for each fixed $L$, the rule $\psi^L$ is resource-monotonic, consistent, and continuous; hence, it is a Young rule.

The award distribution of the \emph{$\alpha_{\min}$-egalitarian rule} introduced by \citet{gimenez2014proportional} coincides in our framework with that of the rule $\psi^{c_1}(\mathbf{c}, e)$, where $c_1$ denotes the smallest claim. The $\alpha_{\min}$ rule does not satisfy the property of \emph{consistency}. The reason is that its fixed component is determined endogenously—by the smallest claim in the problem—which may change when considering subproblems. In contrast, the rule $\psi^{c_1}$ treats the fixed component as exogenously specified, ensuring consistency across subsets. Nevertheless, if the agent with the smallest claim remains in the subset under consideration, the $\alpha_{\min}$ rule behaves consistently. To formalize this, we introduce a stronger notion of consistency.

\begin{definition}[Consistency$_1$]
Let $N \in \mathbb{N} \setminus \{1\}$, $(\mathbf{c}, e) \in \mathcal{C}$, and let $N^2 = N' \cup \{1\}$ for some $N' \subset N$. If $x \equiv \varphi(\mathbf{c}, e)$, then
\[
\varphi\!\left(\mathbf{c}|_{N^2}, \sum_{i \in N^2} x_i \right) = x|_{N^2}.
\]
\end{definition}

The $\alpha_{\min}$ rule satisfies Consistency$_1$ whenever the agent with the smallest claim remains in the subset under consideration. For the same reason, it also does not satisfy \emph{converse consistency}. Table~\ref{tab:yourlabel} summarizes the fairness properties satisfied by the most relevant division rules considered in this paper.

\begin{table}[ht]
\centering
\caption{Fairness properties satisfied by selected division rules}
\vspace{0.3cm}
\label{tab:yourlabel}
\begin{tabular}{lcccc}
\toprule
\textbf{Property} & \textbf{Proportional} & \textbf{CEA} & $\boldsymbol{\alpha_{\min}}$ & \textbf{P–CEA Family} \\
\midrule
Equal Treatment of Equals & yes & yes & yes & yes \\
Order Preservation         & yes & yes & yes & yes \\
Resource Monotonicity      & yes & yes & yes & yes \\
Anonymity                  & yes & yes & yes & yes \\
Consistency                & yes & yes & no  & yes \\
Converse Consistency       & yes & yes & no  & yes \\
Composition Down           & yes & yes & yes & yes \\
Composition Up             & yes & yes & yes & yes \\
Young Rule Representation  & yes & yes & no  & yes \\
\bottomrule
\end{tabular}
\end{table}

\begin{example}
To illustrate the behavior of the different allocation rules, consider the claims problem with $\mathbf{c} = (10, 50, 70)$ and estate $e = 100$. Table~\ref{tab:yourlabel1} reports the allocations generated by several benchmark and compromise rules. In this case, the largest admissible value of the parameter is $\lambda = 45$. The proportional rule corresponds to $L = 0$, while the CEA rule corresponds to $L = 45$.

\begin{table}[ht]
\centering
\caption{Allocations for claims $(10, 50, 70)$ and estate $e = 100$}
\vspace{0.3cm}
\label{tab:yourlabel1}
\begin{tabular}{lccc}
\toprule
\textbf{Division Rule} & \textbf{Agent 1} & \textbf{Agent 2} & \textbf{Agent 3} \\
\midrule
Proportional           & 7.692            & 38.462           & 53.846           \\
CEA                    & 10               & 45               & 45               \\
CEL                    & 0                & 40               & 60               \\
$\alpha_{\min}$        & 10               & 38               & 52               \\
$\psi^{L=5}$           & 8.696            & 38.261           & 53.043           \\
$\psi^{L=20}$          & 10               & 38.75            & 51.25            \\
$\psi^{L=40}$          & 10               & 42.5             & 47.5             \\
\bottomrule
\end{tabular}
\end{table}

The example reveals substantial variation across allocation methods. The CEL rule produces the most unequal outcome, assigning nothing to Agent~1 and the largest share to Agent~3. By contrast, the CEA rule yields the most egalitarian allocation, equalizing the awards of Agents~2 and~3 while fully satisfying Agent~1’s claim. 

Within the P--CEA family, allocations become progressively more egalitarian as $L$ increases. A higher value of $L$ raises the guaranteed baseline component and reduces the weight placed on proportionality, thereby shifting the distribution toward equal awards.
\end{example}

\subsection{Inequality Comparisons}

We now assess the distributive implications of the P--CEA family using standard inequality orderings. These criteria allow us to compare the rules not only in terms of fairness axioms, but also in terms of the degree of equality they generate across agents.

\subsubsection{Leximin Ordering}

Let \( x = (x_1, \ldots, x_n) \) and \( y = (y_1, \ldots, y_n) \) be two vectors in \( \mathbb{R}^n \), sorted in non-decreasing order. We say that \( x \) is preferred to \( y \) under the \emph{leximin ordering}, denoted \( x \succ_{\text{lex}} y \), if there exists an index \( m \in \{0, 1, \ldots, n-1\} \) such that
\[
x_i = y_i \text{ for all } i = 1, \ldots, m, 
\quad \text{and} \quad 
x_{m+1} > y_{m+1}.
\]

The leximin criterion gives absolute priority to the worst-off agent, breaking ties by comparing the second worst-off, and so on. Among all division rules, the CEA rule maximizes the leximin ordering over awards, while the CEL rule maximizes the leximin ordering over losses.

The P--CEA family $\psi^L$ lies between the two benchmark rules in the leximin sense:
\[
CEA(\mathbf{c},e) \succ_{\text{lex}} \psi^L(\mathbf{c},e) \succ_{\text{lex}} P(\mathbf{c},e).
\]
Moreover, for any \( L, L' \in [0, \lambda] \) with \( L' > L \),
\[
\psi^{L'}(\mathbf{c},e) \succ_{\text{lex}} \psi^L(\mathbf{c},e).
\]
Thus, increasing $L$ monotonically improves the allocation in the leximin ordering.

\subsubsection{Lorenz Domination}

Let \( x = (x_1, \ldots, x_n) \) and \( y = (y_1, \ldots, y_n) \) be vectors in \( \mathbb{R}^n \), sorted in non-decreasing order. We say that \( x \) \emph{Lorenz-dominates} \( y \), written \( x \succcurlyeq_{\text{LD}} y \), if
\begin{itemize}
    \item \( \sum_{i=1}^k x_i \geq \sum_{i=1}^k y_i \quad \text{for all } k = 1, \ldots, n-1, \)
    \item and \( \sum_{i=1}^n x_i = \sum_{i=1}^n y_i \).
\end{itemize}

Lorenz domination captures the idea that \( x \) is more equally distributed than \( y \). As shown by \citet{bosmans2011lorenz}, for any division rule \( \varphi \) and any claims problem \( (\mathbf{c},e) \),
\[
CEA(\mathbf{c},e) \succcurlyeq_{\text{LD}} \varphi(\mathbf{c},e) \succcurlyeq_{\text{LD}} CEL(\mathbf{c},e).
\]

For the P--CEA family, we have
\[
CEA(\mathbf{c},e) \succcurlyeq_{\text{LD}} \psi^L(\mathbf{c},e) \succcurlyeq_{\text{LD}} P(\mathbf{c},e).
\]
More precisely, for any \( L \in [0,\lambda] \),
\begin{equation*}
\left\{
\begin{aligned}
    &CEA(\mathbf{c},e) \succcurlyeq_{\text{LD}} \alpha_{\min}(\mathbf{c},e) \succcurlyeq_{\text{LD}} \psi^L(\mathbf{c},e) \succcurlyeq_{\text{LD}} P(\mathbf{c},e), && \text{if } L \leq c_1; \\
    &CEA(\mathbf{c},e) \succcurlyeq_{\text{LD}} \psi^L(\mathbf{c},e) \succcurlyeq_{\text{LD}} \alpha_{\min}(\mathbf{c},e) \succcurlyeq_{\text{LD}} P(\mathbf{c},e), && \text{otherwise}.
\end{aligned}
\right.
\end{equation*}

These comparisons show that the rule $\psi^L$ constitutes a genuine compromise between the CEA and proportional benchmarks. As the parameter $L$ increases, the guaranteed component expands and the allocation shifts toward greater equality in the Lorenz sense, while still preserving proportional sensitivity to claims. The next result formalizes this monotonic relationship between the parameter $L$ and Lorenz dominance.

\begin{proposition}
For any \( L, L' \in [0, \lambda] \) with \( L' > L \), we have
\[
\psi^{L'}(\mathbf{c}, e) \succcurlyeq_{LD} \psi^L(\mathbf{c}, e).
\]
\end{proposition}

\begin{proof}
Let the claims vector be ordered as
\[
c_1 \leq c_2 \leq \dots \leq c_k \leq L \leq c_{k+1} \leq \dots \leq c_n, \text{ or }
\]
\[
L\leq c_1\leq ... \leq c_n.
\]
We show that for $L^\prime > L$ the award vector under \( \psi^{L'} \) Lorenz-dominates that under \( \psi^L \).

Since the total estate \( e \) is fixed, it suffices to compare cumulative sums of awards. We proceed sequentially from the smallest claim upward. To begin with, we show, at the first index where the two allocations differ, the award under \( \psi^{L'} \) is strictly larger.

\textit{Case 1:} \( c_k \leq L \leq c_{k+1} < L' \).

For all agents \( i \leq k \), both rules award the full claim \( c_i \), so the allocations coincide. For agent \( k+1 \), under \( \psi^L \),
\[
\psi^L_{k+1}
=
L + P(\mathbf{c} - \mathbf{c}^L, e - e^L)_{k+1},
\]
whereas under \( \psi^{L'} \), since \( L' > c_{k+1} \), the fixed component equals the full claim:
\[
\psi^{L'}_{k+1} = c_{k+1}.
\]
Hence \( \psi^{L'}_{k+1} > \psi^L_{k+1} \), and the cumulative sum up to agent \( k+1 \) is strictly larger under \( \psi^{L'} \).

\textit{Case 2:} \( c_k \leq L < L' \leq c_{k+1} \).

Again, for agents \( i \leq k \), awards are identical under both rules. For agent \( k+1 \),
\[
\psi^L_{k+1}
=
L + P(\mathbf{c} - \mathbf{c}^L, e - e^L)_{k+1},
\quad
\psi^{L'}_{k+1}
=
L' + P(\mathbf{c} - \mathbf{c}^{L'}, e - e^{L'})_{k+1}.
\]

Let \( L' = L + \varepsilon \) with \( \varepsilon > 0 \). Then
\[
\mathbf{c}^{L'}
=
(\min\{c_i, L + \varepsilon\})_{i \in N},
\quad
e^{L'}
=
e - \sum_{j \leq k} c_j - (n-k)(L + \varepsilon).
\]

Consider the difference:
\begin{align*}
\psi^{L'}_{k+1} - \psi^L_{k+1}
&=
\left[
L + \varepsilon
+
\frac{c_{k+1} - L - \varepsilon}
{c_N - \sum_i c_i^L - (n-k)\varepsilon}
\big(e - e^L - (n-k)\varepsilon\big)
\right] \\
&\quad
-
\left[
L
+
\frac{c_{k+1} - L}
{c_N - \sum_i c_i^L}
(e - e^L)
\right] \\
&=
\varepsilon
\left[
\frac{c_N - \sum_i c_i^L - (e - e^L)}
{c_N - \sum_i c_i^L - (n-k)\varepsilon}
\right] \\
&\quad
+
\frac{(c_{k+1} - L)(n-k)\varepsilon
\big[(e - e^L) - (c_N - \sum_i c_i^L)\big]}
{(c_N - \sum_i c_i^L - (n-k)\varepsilon)
(c_N - \sum_i c_i^L)} \\
&=
\varepsilon
\left[
\frac{c_N - \sum_i c_i^L - (e - e^L)}
{c_N - \sum_i c_i^L - (n-k)\varepsilon}
\right]
\cdot
\left[
\frac{\sum_{i>k}(c_i - c_k)}
{c_N - \sum_i c_i^L}
\right]
\geq 0.
\end{align*}

Thus, the award to agent \( k+1 \) strictly increases as \( L \) rises to \( L' \), which increases the cumulative sum at the first index of difference.

\textit{Case 3:} \( L \leq c_1 \leq L' \).

In this case,
\[
\psi^L_{1}
=
L + P(\mathbf{c} - \mathbf{c}^L, e - e^L)_{1},
\quad
\psi^{L'}_{1} = c_1.
\]
Since \( c_1 > \psi^L_{1} \), the first component strictly increases.

\bigskip 

These three cases exhaust all possible positions of $L^\prime$ relative to the claims and $L$. To show the monotonicity of cumulative differences, notice for an agent \( k+i \) where $i\in\{1,...,n-k\}$,
\[
\psi^L_{k+i}
=
L + P(\mathbf{c} - \mathbf{c}^L, e - e^L)_{k+i},
\quad
\psi^{L'}_{k+i}
=
L' + P(\mathbf{c} - \mathbf{c}^{L'}, e - e^{L'})_{k+i}.\]

Following Case 2, we obtain
\begin{align*}
    &\psi^{L'}_{k+i} - \psi^{L}_{k+i} \ge 0\\
    i.e., &\psi^{L'}_i - \psi^{L}_i \ge 0 \quad \forall i\ge k+1
\end{align*}
Hence, from the first index at which the award under $\psi^{L'}$ exceeds that under $\psi^{L}$, the cumulative difference between the two allocations remains non-negative. Adding further components cannot reverse the inequality of the cumulative sums. Therefore, the cumulative sums of $\psi^{L'}$ are weakly larger than those of $\psi^{L}$ at every index, implying that
\[
\psi^{L'}(\mathbf{c}, e) \succcurlyeq_{LD} \psi^{L}(\mathbf{c}, e).
\]
\end{proof}

\section{Axiomatic Characterizations}

We now introduce the main axioms used to characterize the proposed P--CEA family of rules. These axioms are adapted from standard principles in the literature and modified to accommodate the presence of an exogenously specified lower bound $L$. Assume $k$ to be the highest claimant with a claim equal to or below $L$ and we write $S_0 = \{1,2,\dots,k\}$ as the set of agents with claims below L.

\begin{axiom}[No Advantageous Reallocation Beyond Claims Below $L$ (NAR$_L$)]
Fix an L. For each $(\mathbf{c}, e) \in \mathcal{C}$ and each subset $S \subseteq N \setminus S_0$, suppose $\{c_i'\}_{i \in S} \in \mathbb{R}_+^S$ satisfies $c_i' \geq L$ for all $i \in S$ and $\sum_{i \in S} c_i' = \sum_{i \in S} c_i$. Then:
\[
\sum_{i \in S} \varphi_i\big(((c_i')_{i \in S}, c_{N \setminus S}), e\big)
=
\sum_{i \in S} \varphi_i(\mathbf{c}, e).
\]
\end{axiom}

The axiom $NAR_L$ requires that no coalition of agents whose claims exceed $L$ can benefit by redistributing their claims among themselves while keeping their aggregate claim unchanged and preserving the lower bound $L$ for each member. It is a restricted version of Moulin’s \emph{No Advantageous Reallocation} (NAR) axiom \citep{moulin1985egalitarianism}. The restriction reflects the role of $L$ as a subsistence or guaranteed minimum level: only agents whose claims exceed this threshold are permitted to reallocate claims within their group.

Intuitively, the axiom protects agents with relatively low claims. By preventing coalitions of higher-claim agents from manipulating their internal distribution of claims, it limits strategic exaggeration and ensures that the guaranteed minimum component is not undermined through coordinated misreporting.

\begin{remark}
An equivalent formulation of $NAR_L$ is \emph{Non-subsistence Decentralizability}. The proof of the same is discussed later in this section. A division rule is said to be non-subsistence decentralizable if, for each agent $i$, her award depends only on her own claim $c_i$, the total claim $c_N$, the available estate $e$, and the vector of subsistence claims $\mathbf{c}_{\text{sub}} = (c_1, \dots, c_k)$. Formally,
\[
\varphi_i(\mathbf{c}, e)
=
t_i(c_i; c_N; e; \mathbf{c}_{\text{sub}}).
\]
\end{remark}

This axiom expresses a natural decentralization requirement: once the subsistence claims are fixed, an agent’s allocation depends only on her own claim and aggregate information. In particular, it eliminates the need for each agent’s award to depend on the entire claim vector. Such informational independence is especially desirable in large populations, where complete knowledge of all claims may be costly or impractical.

A useful implication of decentralizability is that the joint award to any non-empty coalition $S \subset N$ can be expressed through a function $r_S$:
\[
\sum_{i \in S} \varphi_i(\mathbf{c}, e)
=
r_S\left(c_S; \sum_{j \in N \setminus S} c_j; e; \mathbf{c}_{\text{sub}}\right).
\]
Thus, a coalition’s total award depends only on its aggregate claim, the aggregate claim of its complement, the available estate, and the subsistence claims.

The second axiom formalizes the idea that the allocation rule must respect a guaranteed subsistence level. Since the parameter $L$ represents an exogenously specified minimum entitlement, the rule should ensure that no agent whose claim exceeds this threshold receives less than $L$. At the same time, agents with claims below $L$ should receive their claims in full, in accordance with claim-boundedness. The following axiom captures this requirement.

\begin{axiom}[Sustainable Lower Bound on Awards (SLBA$_L$)]
Fix an L. For each $(\mathbf{c}, e) \in \mathcal{C}$ and each $i \in N$, we require:
\[
\varphi_i(\mathbf{c}, e) \geq \min\{c_i, L\}.
\]
\end{axiom}

The axiom $SLBA_L$ guarantees that every agent receives at least the minimum of her claim and the exogenously specified threshold $L$. Thus, an agent whose claim exceeds $L$ is assured a baseline award of $L$, while an agent with a smaller claim receives her claim in full.

This requirement strengthens claim-boundedness by introducing a guaranteed floor, yet it remains weaker than the \emph{fair lower bound} condition of \citet{moulin1992welfare}, which uses $\min\{c_i, e/n\}$ as the benchmark. The two bounds coincide only when $L = e/n$; otherwise, $SLBA_L$ imposes a strictly less demanding constraint. Interpreting $L$ as a subsistence threshold—such as a minimum calorie requirement—the axiom states that an agent can receive less than $L$ only if her claim itself falls below $L$. In that case, the shortfall reflects that her essential needs are already met through external means, rather than through the allocation mechanism.

\begin{lemma}
\label{lem:pinned}
    For all $i\leq k$, $SLBA_L$ and claims-boundedness together imply $\varphi_i(\mathbf{c},e)=c_i$.
\end{lemma}

\begin{proof}
Fix $(\mathbf{c}, e) \in \mathcal{C}$ and let $i \leq k$, so that $c_i \leq L$ by definition of $k$. Claim-boundedness requires that no agent receives more than her claim, yielding $\varphi_i(\mathbf{c}, e) \leq c_i$.

On the other hand, \textup{SLBA}$_L$ guarantees that each agent receives at least $\min\{c_i, L\}$. Since $c_i \leq L$, this minimum equals $c_i$, so that $\varphi_i(\mathbf{c}, e) \geq \min\{c_i, L\} = c_i$.

The two inequalities together yield $\varphi_i(\mathbf{c}, e) = c_i$, as claimed.
\end{proof}

The above lemma shows all the agents with claims equal to or below L are pinned to their claims by $SLBA_L$ and feasibility. Also for all $i\geq k+1$, we have $\varphi_i(\mathbf{c},e)\geq L$, meaning all agents with claims above $L$ are ensured the amount $L$.

We now provide the main characterization of the P--CEA family. 
The result shows that combining protection against strategic reallocation 
with a guaranteed lower bound uniquely determines the compromise structure.

\begin{theorem}
\label{th1}
A rule satisfies the axioms of \textit{$NAR_L$} and \textit{$SLBA_L$} if and only if it belongs to the P--CEA family of compromise rules.
\end{theorem}

\begin{proof}
We prove the statement in two parts. First, we establish the necessary condition, namely that every rule in the P--CEA family satisfies the axioms \textit{$NAR_L$} and \textit{$SLBA_L$}.

\textit{Necessary condition.}
The compromise rules of the P--CEA family satisfy the axiom $NAR_L$. The resource is divided according to
\[
\psi^L(\mathbf{c},e)=\mathbf{c}^L+P(\mathbf{c}-\mathbf{c}^L,e^L).
\]
Here, $P(\mathbf{c}-\mathbf{c}^L, e^L)$ is strictly positive only for agents $\{k+1,\dots,n\}$. Suppose $i\in S$ reports $c_i'$ such that $\sum_{i\in S}c_i' = \sum_{i\in S}c_i$. Then
\begin{align*}
\sum_{i\in S}\varphi_i(((c_i')_{i\in S},c_{N\setminus S}),e)
&= \sum_{i\in S} L + \sum_{i\in S}\frac{(c_i' - L) \cdot e^L}{c_N - \sum_{i\in N} c^L_i} \\
&= \sum_{i\in S} L + \sum_{i\in S}\frac{(c_i - L) \cdot e^L}{c_N - \sum_{i\in N} c^L_i} \\
&= \sum_{i\in S} \varphi_i(\mathbf{c},e).
\end{align*}
Thus, $NAR_L$ is satisfied.

The P--CEA family also satisfies $SLBA_L$. For agents $1,2,\dots,k$, since $c_i \leq L$, they are fully awarded, i.e., $\varphi_i = c_i$ for all $i\in \{1,2,\dots,k\}$. For $i\in \{k+1,\dots,n\}$, we have $c_i > L$ and
\[
\varphi_i(\mathbf{c},e) = L + P_i(\mathbf{c}-\mathbf{c}^L,e^L) > L.
\]

\textit{Sufficiency condition.}
Now consider a rule $\varphi(\mathbf{c}, e)$ that satisfies the two axioms. We show that this rule coincides with $\psi^L(\mathbf{c},e)$.
We can write the claim vector as
\[
\mathbf{c}=(c_1,\dots,c_k,L+c_{k+1}^\prime,L+c_{k+2}^\prime,\dots,L+c_n^\prime),
\]
where $c_{k+i}^\prime=c_{k+i}-L$. From Lemma~\ref{lem:pinned} we know that for all $i\leq k$, $\varphi_i(\mathbf{c},e)=c_i$ and for all $i\geq k+1$, $\varphi_i(\mathbf{c},e)\geq L$.

If we transfer the claims excess of amount $L$ of the $(k+2)^\text{th}$ agent to the previous agent, the axiom $NAR_L$ guarantees that the two agents $\{k+1,k+2\}$ together do not benefit. Hence,
\begin{align}
& \varphi_{k+1}(\mathbf{c},e) + \varphi_{k+2}(\mathbf{c},e) \notag \\
= & \varphi_{k+1}((c_1,\dots,c_k,L+c_{k+1}^\prime+c_{k+2}^\prime,L,L+c_{k+3}^\prime,\dots,L+c_n^\prime),e) + L.
\end{align}

Again, by applying $NAR_L$, when the subgroup $\{k+2,\dots,n\}$ transfers claims among themselves, agent $(k+1)$ is unaffected. This yields
\begin{align}
\varphi_{k+1}(\mathbf{c},e)
&= \varphi_{k+1}((c_1,\dots,c_k,L+c_{k+1}^\prime, c_N - \sum_{i=1}^k c_i - (n-k)L - c_{k+1}^\prime, L,\dots,L),e) \notag \\
&= \varphi_{k+1}((c_1,\dots,c_k,L+c_{k+1}^\prime,L, c_N - \sum_{i=1}^k c_i - (n-k)L - c_{k+1}^\prime, L,\dots,L),e).
\end{align}

Next, consider the coalition of the agents $\{k+1,k+3,\dots,n\}$. Since agent $(k+2)$ is not part of this group, her award remains unchanged. Then, consider the coalition $\{k+1, k+2\}$. By $NAR_L$, their joint award is invariant to any reallocation of claims within the coalition that preserves the aggregate and keeps each claim at or above $L$. In particular, transferring the entire excess of agent $(k+2)$ to agent $(k+1)$ leaves the joint award unchanged.
\begin{align}
\varphi_{k+2}(\mathbf{c},e)
&= \varphi_{k+2}((c_1,\dots,c_k,L,L+c_{k+2}^\prime, c_N - \sum_{i=1}^k c_i - (n-k)L - c_{k+2}^\prime, L,\dots,L),e) \notag \\
&= \varphi_{k+1}((c_1,\dots,c_k,L+c_{k+2}^\prime,L, c_N - \sum_{i=1}^k c_i - (n-k)L - c_{k+2}^\prime, L,\dots,L),e).
\end{align}

The last equality follows from the fact that when an agent claims $L$, $SLBA_L$ assures her of L and claim boundedness caps her award at $L$. These two conditions together imply that the agent receives an award of $L$, hence -
\begin{align*}
& \varphi_{k+1}((c_1,\dots,c_k,L,L+c_{k+2}^\prime, c_N - \sum_{i=1}^k c_i - (n-k)L - c_{k+2}^\prime, L,\dots,L),e) =L\\
\text{and }\quad & \varphi_{k+2}((c_1,\dots,c_k,L+c_{k+2}^\prime,L, c_N - \sum_{i=1}^k c_i - (n-k)L - c_{k+2}^\prime, L,\dots,L),e)=L.
\end{align*}
Thus, the agent pinned to a claim of $L$ necessarily receives $L$, and the 
other agent absorbs the remainder of the joint award.

Similarly,
\begin{align}
& \varphi_{k+1}((c_1,\dots,c_k,L+c_{k+1}^\prime+c_{k+2}^\prime,L, L+c_{k+3}^\prime,\dots,L+c_n^\prime),e) \notag \\
= & \varphi_{k+1}((c_1,\dots,c_k,L+c_{k+1}^\prime+c_{k+2}^\prime,L, c_N - \sum_{i=1}^k c_i - (n-k)L - (c_{k+1}^\prime + c_{k+2}^\prime), L,\dots,L),e).
\end{align}

Substituting (2), (3), and (4) into (1), we obtain
\begin{align*}
& \varphi_{k+1}((c_1,\dots,c_k,L+c_{k+1}^\prime,L, c_N - \sum_{i=1}^k c_i - (n-k)L - c_{k+1}^\prime, L,\dots,L),e) \\
+ & \varphi_{k+1}((c_1,\dots,c_k,L+c_{k+2}^\prime,L, c_N - \sum_{i=1}^k c_i - (n-k)L - c_{k+2}^\prime, L,\dots,L),e) \\
= & \varphi_{k+1}((c_1,\dots,c_k,L+c_{k+1}^\prime+c_{k+2}^\prime,L, c_N - \sum_{i=1}^k c_i - (n-k)L - (c_{k+1}^\prime + c_{k+2}^\prime), L,\dots,L),e). \tag{5}
\end{align*}

Only the claim of agent $(k+1)$ varies in these expressions. Let us define $\psi:[L,L+c_N^\prime]\rightarrow \mathbb{R}$ as -
\[
\psi(L+t)=\varphi_{k+1}((c_1,\dots,c_k,L+t,L, c_N - \sum_{i=1}^k c_i - (n-k)L - t, L,\dots,L),e).
\]
We can rewrite equation (5) as -
\[
\psi(L+c_{k+1}^\prime) + \psi(L+c_{k+2}^\prime)
=
\psi(L+c_{k+1}^\prime+c_{k+2}^\prime) + L.
\]
Notice that the above equation is a variant of Cauchy's functional equation (see \cite{aczel1966lectures}). Notice that the function $\psi$ is defined on the bounded interval $[L,L+c_N^\prime]$ and a fair division rule satisfies the conditions on non-negativity, claim-boundedness, and budget-balancedness. These domain restrictions are sufficient to rule out pathological (non-measurable, everywhere dense) solutions to the Cauchy equation as the Cauchy equation that is bounded on any interval of positive measure must be linear. We write for any $\psi:[L,L+c_N^\prime]\rightarrow \mathbb{R}$ and $t\in[0,c_N^\prime]$, there exists an $m\in \mathbb{R}_+$ such that
\[
\psi(L+t) = mt+ L,
\]
where $c_N^\prime=\sum_{i=k+1}^nc_i^\prime$. That is,
\[
\varphi_i(\mathbf{c},e) = m c_i^\prime + L.
\]

Budget balancedness implies
\begin{align*}
\sum_{i=k+1}^{n} \varphi_i(\mathbf{c},e)
&= \sum_{i=k+1}^{n} [m c_i^\prime + L] \\
&= m\left(c_N - \sum_{i=1}^k c_i - (n-k)L\right) + (n-k)L \\
&= e^L,
\end{align*}
so that
\[
m = \frac{e^L}{c_N - \sum_{i=1}^k c_i - (n-k)L}.
\]

Substituting back,
\[
\varphi_i(\mathbf{c},e)
=
L + (c_i - L)\left(\frac{e^L}{c_N - \sum_{i=1}^k c_i - (n-k)L}\right),
\quad \forall i \geq k+1.
\]

Therefore,
\begin{align*}
\varphi_i(\mathbf{c},e)
&= \min\{c_i,L\} + (c_i - L)\left(\frac{e^L}{c_N - \sum_{i=1}^k c_i - (n-k)L}\right) \\
&= c_i^L + P_i(\mathbf{c} - \mathbf{c}^L, e^L) \\
&= \psi^L_i(\mathbf{c},e),
\end{align*}
which coincides with the P--CEA family of compromise rules, where $c_i^L = \min\{c_i, L\}$ and $e^L = e - \sum_{i=1}^k c_i - (n-k)L$.

\end{proof}

\begin{remark}
The axioms \textit{NAR}$_L$ and \textit{SLBA}$_L$ are logically independent. 
The proportional rule satisfies \textit{NAR}$_L$ but violates \textit{SLBA}$_L$, 
whereas the Constrained Equal Awards (CEA) rule satisfies \textit{SLBA}$_L$ 
but fails to satisfy \textit{NAR}$_L$.
\end{remark}

Now we want to show that Non-subsistence Decentralizability, along with $SLBA_L$, also characterizes the $P-CEA$ family of rules. First, we show that $NAR_L$ and Non-subsistence Decentralizability are the same axioms.
\begin{lemma}
    The axiom of $NAR_L$ is equivalent to the axiom of Non-subsistence Decentralizability.
\end{lemma}

\begin{proof}

\noindent ($\Leftarrow$) Non-subsistence Decentralizability implies $NAR_L$.

\noindent For all $S\subseteq N$, Non-subsistence Decentralizability means $$\sum_{i\in S}\varphi_i(\mathbf{c},e)=r_S(\sum_{i\in S} c_i;\sum_{i\in N\setminus S}c_i;e;c_{sub})\,.$$

Let $S_1\subset N\setminus S_0$, and suppose the agents in $S_1$ change their claims such that $\sum_{i\in S_1}c_i=\sum_{i\in S_1}c_i^\prime$. We show that the aggregate award received by agents in $S_1$ remains unchanged. Indeed,
\begin{align*}
    \sum_{i\in S_1}\varphi_i(\mathbf{c},e) 
    &= r_{S_1}\!\left(\sum_{i\in S_1}c_i;\sum_{i\in N\setminus S_1}c_i;e;c_{sub}\right) \\
    &= r_{S_1}\!\left(\sum_{i\in S_1}c_i^\prime;\sum_{i\in N\setminus S_1}c_i;e;c_{sub}\right) \\
    &= \sum_{i\in S_1}\varphi_i(((c_i^\prime)_{i\in S_1},(c_i)_{i\in N\setminus S_1}),e),
\end{align*}

where the second equality uses the fact that the aggregate claim of $S_1$ is unchanged. Hence $\varphi$ satisfies $NAR_L$.

\medskip
\noindent ($\Rightarrow$) $NAR_L$ implies Non-subsistence Decentralizability.

Let $S_2\subset N$ be arbitrary, and write $S_2=S_{20}\cup S_{21}$, where $S_{20}\subset S_0$ so that $(c_i)_{i\in S_{20}}\subset c_{sub}$, and $S_{21}\subset N^1$. Note that $NAR_L$ applies only to subsets of $N^1$.

We first establish that $\sum_{i\in S_2}\varphi_i(\mathbf{c},e)$ is insensitive to redistributions of the aggregate claim of $N^1\setminus S_{21}$. Applying $NAR_L$ to $N^1\setminus S_{21}$ gives
\begin{align*}
    \sum_{i\in S_2}\varphi_i(\mathbf{c},e)
    &= e - \sum_{i\in N^1 \setminus S_{21}}\varphi_i(\mathbf{c},e)
    - \sum_{i\in S_0\setminus S_{20}}\varphi_i(\mathbf{c},e) \\
    &= e - \sum_{i\in N^1 \setminus S_{21}}\varphi_i(((c_i)_{i\in S_{21}\cup S_0},
    (c_i^\prime)_{i\in N^1\setminus S_{21}}),e)
    - \sum_{i\in S_0\setminus S_{20}}\varphi_i(\mathbf{c},e) 
    \end{align*}
    where $\sum_{i\in N^1 \setminus S_{21}}c_i
    = \sum_{i\in N^1 \setminus S_{21}}c_i^\prime$. Hence we can write -
    \begin{align*}
    \sum_{i\in S_2}\varphi_i(\mathbf{c},e)= \sum_{i\in S_2}\varphi_i(((c_i)_{i\in S_{21}},
    (c_i)_{i\in S_0},(c_i^\prime)_{i\in N^1\setminus S_{21}}),e).
\end{align*}
This shows that so long as $\sum_{i\in N^1 \setminus S_{21}}c_i$ is held constant, the aggregate award of $S_2$ is unaffected by how that total is distributed within $N^1\setminus S_{21}$.

We next establish the analogous insensitivity with respect to $S_{21}$. Applying $NAR_L$ to $S_{21}$ yields
\begin{align*}
    \sum_{i\in S_2}\varphi_i(\mathbf{c},e)
    &= \sum_{i\in S_{20}}\varphi_i(((c_i^\prime)_{i\in S_{21}},
    (c_i)_{i\in N\setminus S_{21}}),e)
    + \sum_{i\in S_{21}}\varphi_i(((c_i^\prime)_{i\in S_{21}},
    (c_i)_{i\in N\setminus S_{21}}),e)
    \end{align*}
    where $\sum_{i\in S_{21}}c_i=\sum_{i\in S_{21}}c_i^\prime$. Hence we write -
    \begin{align*}
    \sum_{i\in S_2}\varphi_i(\mathbf{c},e) &= \sum_{i\in S_2}\varphi_i(((c_i^\prime)_{i\in S_{21}},
    (c_i)_{i\in N\setminus S_{21}}),e) \\
    &= \sum_{i\in S_2}\varphi_i(((c_i^\prime)_{i\in S_{21}},
    (c_i)_{i\in S_0},(c_i)_{i\in N^1\setminus S_{21}}),e).
\end{align*}
Hence the aggregate award of $S_2$ is equally insensitive to how the total $\sum_{i\in S_{21}}c_i$ is distributed within $S_{21}$.

Combining both steps, we conclude that $\sum_{i\in S_2}\varphi_i(\mathbf{c},e)$ depends only on $\sum_{i\in S_{21}}c_i$, $\sum_{i\in N^1 \setminus S_{21}}c_i$, $e$, and the subsistence claim vector $c_{sub} = (c_i)_{i\in S_0}$. Since the full vector $c_{sub}$ is already taken as given, knowledge of $\sum_{i\in S_{21}}c_i$ and $\sum_{i\in N^1 \setminus S_{21}}c_i$ is informationally equivalent to knowledge of $\sum_{i\in S_2}c_i$ and $\sum_{i\in N\setminus S_2}c_i$. Therefore,
\[
    \sum_{i\in S_2} \varphi_i(\mathbf{c},e)
    = r_{S_2}\!\left(\sum_{i\in S_2}c_i;\sum_{i\in N\setminus S_2}c_i;e;c_{sub}\right),
\]
establishing Non-subsistence Decentralizability.
\end{proof}

The characterization of Theorem \ref{th1} can be reformulated using the equivalent 
decentralization property. Since $NAR_L$ and Non-subsistence Decentralizability 
are equivalent, replacing the former with the latter yields the same class of 
rules. This observation leads to the following immediate corollary.

\begin{corollary}
A rule satisfies \emph{Non-subsistence Decentralizability} and 
\emph{Sustainable Lower Bounds on Awards (SLBA$_L$)} 
if and only if it belongs to the P--CEA family of compromise rules.
\end{corollary}

\noindent
This corollary follows directly from the equivalence between $NAR_L$ and 
Non-subsistence Decentralization established above. It highlights a key feature 
of the P--CEA family: once subsistence claims are fixed, each agent’s award is 
determined solely by her own claim and aggregate magnitudes, without requiring 
knowledge of the entire claim vector. The resulting rule is decentralized, 
informationally parsimonious, and straightforward to implement.

The award vector of $\alpha_{\min}$ egalitarian-rule introduced by \cite{gimenez2014proportional} coincides with that of the $\psi^L$ rule when $L = c_1$. This rule is well known and has found wide application in the literature - like carbon budgeting. Formally, \cite{gimenez2014proportional} define the $\alpha_{\min}$egalitarian-rule as follows: for each $(\mathbf{c}, e) \in \mathcal{C}$ and for each $i \in N$,
\[
\varphi_{\min}(\mathbf{c}, e) =
\begin{cases}
\frac{e}{n}\mathbf{1} & \text{if } c_1 \geq \frac{e}{n}, \\
\mathbf{c^1} + P(e - n c_1, \mathbf{c} - \mathbf{c^1}) & \text{otherwise},
\end{cases}
\]
where 
\[
\mathbf{c^1} =
\begin{pmatrix}
c_1 \\
\vdots \\
c_1
\end{pmatrix}
\quad \text{and} \quad
\mathbf{1} =
\begin{pmatrix}
1 \\
\vdots \\
1
\end{pmatrix}.
\]

Under the case $L = c_1$, the two axioms characterizing the $P$-CEA family admit a natural interpretation. The axiom $NAR_{c_1}$ can be understood as follows: no coalition of agents $S \subseteq N \setminus S_0$ can benefit by redistributing their claims among themselves while keeping their total claim fixed and ensuring that no individual claim falls below $c_1$, where $S_0$ denotes the set of agents with the lowest claim. The axiom $SLBA_{c_1}$ guarantees each agent a minimum award equal to the lower of the smallest claim and the equal division benchmark.

The following proposition shows that these two axioms jointly characterize the $\alpha_{\min}$-egalitarian rule. To the best of our knowledge, such an axiomatic characterization has not been established before.

\begin{proposition}
A rule satisfies $NAR_{c_1}$ and \emph{Sustainable Lower Bounds on Awards (SLBA$_{c_1}$)} if and only if it is the $\alpha_{\min}$-egalitarian rule.
\end{proposition}

Despite this equivalence, there is a fundamental conceptual difference between the two rules. In $\psi^{c_1}$, the minimum guaranteed award is exogenously fixed at the level $c_1$. In contrast, under the $\alpha_{\min}$-egalitarian rule, the minimum guarantee is tied to the lowest claim itself. As a result, this benchmark varies across problems as the claims vector changes. This distinction has important implications: while $\psi^{c_1}$ satisfies consistency, the $\alpha_{\min}$-egalitarian rule generally fails to do so.

\section{Dual Analysis}

A \emph{claims problem} can be examined from two complementary viewpoints.  
The standard perspective focuses on the available resources (the \emph{estate}).  
An alternative perspective centers on the \emph{deficit}, defined as the shortfall 
$c_N - e$ between total claims and the available estate.  

Two problems are said to be \emph{dual} if they share the same claims vector 
$\mathbf{c}$ and the estate in one problem equals the deficit in the other.

\begin{definition}[Dual of a rule $S$]
For any $(\mathbf{c},e)\in \mathcal{C}$, the \emph{dual rule} of $S$, denoted $S^d$, is defined by
\[
S^d(\mathbf{c},e) = \mathbf{c} - S(\mathbf{c}, c_N - e),
\]
where $c_N = \sum_{i \in N} c_i$.
\end{definition}

A rule is called \emph{self-dual} if 
\[
S(\mathbf{c},e) = \mathbf{c} - S(\mathbf{c}, c_N - e),
\]
that is, if its award path is symmetric with respect to $\frac{\mathbf{c}}{2}$.  
The notion of duality also applies to axioms: two properties are said to be 
\emph{dual} if, whenever a rule satisfies one of them, its dual satisfies the other.

\medskip

We now introduce the dual of the $\psi^L$ rule, denoted $\psi_{\text{dual}}^L$, 
which provides a compromise between the proportional rule and the 
\emph{equal losses} method.  

Under the equal losses (EL) method, each agent incurs an identical loss 
$\frac{c_N - e}{n}$, yielding
\[
EL_i(\mathbf{c}, e) = c_i - \frac{c_N - e}{n}, 
\quad \forall i \in N.
\]
When the estate is insufficient, losses can be interpreted as consisting of two components: a common (equal) part and a residual part determined proportionally.

Recall that $\mathbf{c}^L = (\min\{c_i, L\})_{i \in N}$.  
We now formally define the dual rule.

\begin{definition}[Dual of the P--CEA family of compromise rules]
For each $(\mathbf{c},e)\in \mathcal{C}$ with $c_i > 0$ for all $i \in N$, define
\begin{align*}
    \psi_{\text{dual}}^L(\mathbf{c},e)
    &= \mathbf{c}-\mathbf{c}^L-P(\mathbf{c}-\mathbf{c}^L,c_N-e-\sum_{i \in N} c^L_i) \\
    &= P(\mathbf{c}-\mathbf{c}^L,e).
\end{align*}
\end{definition}

This dual rule reflects the principle of \emph{preeminence} introduced by 
\citet{herrero1998preeminence}. When resources are extremely scarce, priority 
is given to agents with larger claims. Each agent first bears a minimum loss 
(ensuring non-negativity of awards), and the remaining estate is then distributed 
proportionally according to the adjusted claims.

A natural interpretation arises in emergency allocation contexts such as medical triage. If one patient faces a life-threatening injury while another presents a mild condition, priority is reasonably given to the more severe case.  When the estate is sufficiently large so that not all agents need to bear a loss equal to the smallest admissible amount, the rule converges toward equalizing losses across agents.

\subsection{Dual Axioms and Characterization}

We now introduce the dual counterpart of the SLBA$_L$ axiom.  
In the deficit framework, since $c_N > e$, every agent must incur some loss.  
The next axiom imposes a lower bound on these losses.

\begin{axiom}[Sustainable Upper Bound on Losses (SUBL$_L$)]
For all $(\mathbf{c},e) \in \mathcal{C}$ and each $i \in N$,
\[
c_i - \varphi_i(\mathbf{c},e) \leq \max\{c_i-L,0\}.
\]
\end{axiom}

\begin{theorem}
A family of rules satisfies the axioms of $NAR_L$ and SUBL$_L$ 
if and only if it is the dual of the P--CEA family of compromise rules.
\end{theorem}

The characterization follows from the duality of the axioms that define the original family. Importantly, $NAR_L$ is \emph{self-dual}; hence, combining $NAR_L$ with SUBL$_L$ characterizes precisely the dual P--CEA family.

Similarly, since \emph{Non-subsistence Decentralizability} is self-dual, 
its combination with the dual lower-bound axiom characterizes the dual 
family of compromise rules. This yields the following corollary.

\begin{corollary}
A family of rules satisfies \emph{Non-subsistence Decentralizability} 
and \emph{SUBL$_L$} 
if and only if it is the dual of the P--CEA family of compromise rules.
\end{corollary}

\section{Conclusion}

This paper introduces the P--CEA family of compromise rules for resolving claims problems in which total claims exceed the available estate. The family reconciles two foundational principles of fair division—proportionality and constrained equal awards—by assigning each agent a fixed, claim-bounded baseline and distributing the remaining estate proportionally to residual claims. Varying the parameter $L$ generates a continuum of rules that smoothly interpolates between the proportional and CEA benchmarks.

We provide an axiomatic characterization of this family based on suitably adapted fairness principles. In particular, we show that the P--CEA rules are precisely those that satisfy \emph{No Advantageous Reallocation Beyond Claims Below $L$} together with \emph{Sustainable Lower Bounds on Awards}. These axioms jointly capture two normative concerns: discouraging strategic inflation of claims while guaranteeing a protected minimum level of allocation. The resulting characterization establishes the P--CEA family as a principled and flexible alternative to classical division rules, capable of balancing equity and claim sensitivity in a transparent manner.

We further analyze the dual family of rules, which reallocates losses rather than awards, and characterize it using the corresponding dual axioms. The fact that the same structural principles govern both awards and losses reveals a deeper symmetry underlying claims problems and underscores the robustness of the P--CEA framework.

Beyond its theoretical contribution, the P--CEA family has clear practical relevance. Its two-component structure—a guaranteed minimum combined with proportional adjustment—closely mirrors allocation mechanisms observed in climate policy, income distribution schemes, sports revenue sharing, and the rationing of essential goods. We also compare the family with related rules, including Young rules and the $\alpha_{\min}$-egalitarian rule, and show that P--CEA rules satisfy key normative properties such as consistency, Lorenz dominance, and composition principles.

Overall, the P--CEA family provides a coherent and adaptable approach to fair allocation under scarcity. Future research may extend this framework to dynamic environments, stochastic claims, or settings with network interdependencies, where entitlements evolve over time and fairness considerations interact with strategic behavior and uncertainty.
\newpage

\bibliography{ref}

@article{bosmans2011lorenz,
  title={Lorenz comparisons of nine rules for the adjudication of conflicting claims},
  author={Bosmans, Kristof and Lauwers, Luc},
  journal={International Journal of Game Theory},
  volume={40},
  pages={791--807},
  year={2011},
  publisher={Springer}
}

@article{curiel1987bankruptcy,
  title={Bankruptcy games},
 author={Curiel, Imma J and Maschler, Michael. and Tijs, Stef H},
  journal={Zeitschrift f{\"u}r operations research},
  volume={31},
  pages={A143--A159},
  year={1987},
  publisher={Springer}
}

@article{duro2020allocation,
  title={The allocation of CO2 emissions as a claims problem},
  author={Duro, Juan Antonio and Gim{\'e}nez-G{\'o}mez, Jos{\'e}-Manuel and Vilella, Cori},
  journal={Energy Economics},
  volume={86},
  pages={104652},
  year={2020},
  publisher={Elsevier}
}

@article{gimenez2014proportional,
  title={A proportional approach to claims problems with a guaranteed minimum},
  author={Gim{\'e}nez-G{\'o}mez, Jos{\'e}-Manuel and Peris, Josep E},
  journal={European Journal of Operational Research},
  volume={232},
  number={1},
  pages={109--116},
  year={2014},
  publisher={Elsevier}
}

@article{gimenez2016global,
  title={The global carbon budget: a conflicting claims problem},
  author={Gim{\'e}nez-G{\'o}mez, Jos{\'e}-Manuel and Teixid{\'o}-Figueras, Jordi and Vilella, Cori},
  journal={Climatic change},
  volume={136},
  pages={693--703},
  year={2016},
  publisher={Springer}
}

@article{ju2021fair,
  title={Fair international protocols for the abatement of GHG emissions},
  author={Ju, Biung-Ghi and Kim, Min and Kim, Suyi and Moreno-Ternero, Juan D},
  journal={Energy Economics},
  volume={94},
  pages={105091},
  year={2021},
  publisher={Elsevier}
}

@article{meinshausen2009greenhouse,
  title={Greenhouse-gas emission targets for limiting global warming to 2 C},
  author={Meinshausen, Malte and Meinshausen, Nicolai and Hare, William and Raper, Sarah CB and Frieler, Katja and Knutti, Reto and Frame, David J and Allen, Myles R},
  journal={Nature},
  volume={458},
  number={7242},
  pages={1158--1162},
  year={2009},
  publisher={Nature Publishing Group UK London}
}

@article{moulin1985egalitarianism,
  title={Egalitarianism and utilitarianism in quasi-linear bargaining},
  author={Moulin, Herv{\'e}},
  journal={Econometrica: Journal of the Econometric Society},
  pages={49--67},
  year={1985},
  publisher={JSTOR}
}

@article{moulin2002axiomatic,
  title={Axiomatic cost and surplus sharing},
  author={Moulin, Herve},
  journal={Handbook of Social Choice and Welfare},
  volume={1},
  pages={289--357},
  year={2002},
  publisher={Elsevier}
}

@article{o1982problem,
  title={A problem of rights arbitration from the Talmud},
  author={O'Neill, Barry},
  journal={Mathematical Social Sciences},
  volume={2},
  number={4},
  pages={345--371},
  year={1982},
  publisher={Elsevier}
}

@article{ipcc2014ipcc,
  title={IPCC Fifth Assessment Report—Synthesis Report},
  author={IPCC},
  journal={IPPC Rome, Italy},
  year={2014}
}

@book{thomson2019,
  title={How to divide when there isn't enough},
  author={Thomson, William},
  year={2019},
  publisher={Cambridge University Press}
}

@book{aczel1966lectures,
  title={Lectures on functional equations and their applications},
  author={Acz{\'e}l, J{\'a}nos},
  year={1966},
  publisher={Academic press}
}

@article{aumann1985game,
  title={Game theoretic analysis of a bankruptcy problem from the Talmud},
  author={Aumann, Robert J and Maschler, Michael},
  journal={Journal of economic theory},
  volume={36},
  number={2},
  pages={195--213},
  year={1985},
  publisher={Elsevier}
}

@book{herrero1998preeminence,
  title={Preeminence and sustainability in bankruptcy problems},
  author={Herrero, Carmen and Villar, Antonio},
  year={1998},
  publisher={Instituto Valenciano de Investigaciones Econ{\'o}micas}
}

@article{thomson2015a,
  title={For claims problems, compromising between the proportional and constrained equal awards rules},
  author={William Thomson},
  journal={Economic Theory},
  volume={60},
  pages={495--520},
  year={2015},
  publisher={Springer}
}

@article{thomson2015b,
  title   = {For claims problems, another compromise between the proportional and constrained equal awards rules},
  author  = {Thomson, William},
  journal = {Journal of Dynamics and Games},
  volume  = {2},
  number  = {3--4},
  pages   = {363--382},
  year    = {2015},
  issn    = {2164-6074}
}

@article{thomson2007existence,
  title={On the existence of consistent rules to adjudicate conflicting claims: a constructive geometric approach},
  author={Thomson, William},
  journal={Review of Economic Design},
  volume={11},
  number={3},
  pages={225--251},
  year={2007},
  publisher={Springer}
}

@book{moulin1991axioms,
  title={Axioms of cooperative decision making},
  author={Moulin, Herv{\'e}},
  year={1991},
  publisher={Cambridge University Press}
}

@article{moulin1992welfare,
  title={Welfare bounds in the cooperative production problem},
  author={Moulin, Herv{\'e}},
  journal={Games and Economic Behavior},
  volume={4},
  number={3},
  pages={373--401},
  year={1992},
  publisher={Elsevier}
}
\end{document}